\documentclass{article}
\usepackage{comment} 
\usepackage{ijcai23}
\usepackage{amsmath,amssymb,amsfonts,amsthm}
\usepackage{xcolor}
\usepackage{bm}
\usepackage{times}
\usepackage{soul}
\usepackage{url}
\usepackage[hidelinks]{hyperref}
\usepackage[utf8]{inputenc}
\usepackage[small]{caption}
\usepackage{graphicx}
\usepackage{booktabs}
\usepackage{algorithm}
\usepackage{algorithmic}
\usepackage[switch]{lineno}
\usepackage{cite}

\renewcommand{\epsilon}{\ve}
\def\ve{\varepsilon}

\newcommand{\pr}[2][]{\mbox{Pr}\ifthenelse{\not\equal{}{#1}}{_{#1}}{}\!\left[#2\right]}

\newtheorem{theorem}{Theorem}

\newcommand{\ignore}[1]{}

\newcommand{\bg}[1]{\medskip\noindent{\bf #1}}

\newcommand{\oldbound}[1]{{}}

\title{Online fair division with arbitrary entitlements}

\author{
Kushagra Chatterjee 
\and
Biswadeep Sen \and
Yuhao Wang
\affiliations
National University of Singapore\\
\emails
e0823067@u.nus.edu \and
e0989386@u.nus.edu \and
yuhaowang@u.nus.edu
}

\begin{document}

\maketitle

\begin{abstract}
The division of goods in the online realm poses opportunities and challenges. While innovative mechanisms can be developed, uncertainty about the future may hinder effective solutions. This project aims to explore fair distribution models for goods among agents with arbitrary entitlements, specifically addressing food charity challenges in the real world. Building upon prior work in \cite{aleksandrov2015online}, which focuses on equal entitlements, our project seeks to better understand the proofs of the theorems mentioned in that paper, which currently only provide proof sketches. Our approach employs different proof techniques from those presented in \cite{aleksandrov2015online}.
\end{abstract}

\section{Introduction} \label{sec:intro}
Resource allocation is a crucial issue that society faces, requiring the distribution of limited and often expensive resources among different parties. As a result of environmental, economic, and technological changes, there is growing pressure to allocate resources efficiently.
Fair division problems \cite{brams1996fair} are categorized based on several factors, such as divisible or indivisible goods, centralized or decentralized mechanisms, and cardinal or ordinal preferences. However, these categories fall short of representing the intricacies of many real-world allocation problems.

Specifically, a significant amount of prior research in fair division presumes an offline and fixed problem, where it is assumed that the agents receiving the resources and the resources themselves are predetermined and unchanging.
However, the practical reality frequently differs from these assumptions. Fair division problems usually occur online \cite{walsh2014allocation, walsh2015challenges, aleksandrov2020online}, and the agents, resources to be distributed, or both may not be predetermined and may potentially change over time.
and thus there is a need to develop more sophisticated and realistic models and mechanisms to address these challenges.

Aleksandrov et al \cite{aleksandrov2015online} propose an algorithm for fairly distributing food donations in a food bank by matching donors with recipients. The algorithm considers factors such as food preferences, dietary restrictions, and geographical location to ensure that the distribution is fair and efficient.
This paper \cite{aleksandrov2015online} uses simulations to evaluate the performance of the proposed algorithm and compares it with other fair division algorithms. The authors also provide some theoretical results regarding the fairness of the algorithm.
Beside, \cite{chakraborty2021weighted} 
present and examine novel fairness concepts based on envy for agents who possess weights that measure their entitlements in the offline allocation of indivisible items, where the least weight-adjusted frequent picker is the next agent to pick. There proposed WEF1 algorithm can be computed efficiently using a weighted round-robin algorithm.

\subsection{Problem Formulations}\label{sec: motivation}
The motivation behind the work of Aleksandrov et al \cite{aleksandrov2015online} was the \emph{Food Bank Problem}. 
The food donation process involves donating food to a food bank, which is then distributed to specific charities serving particular groups of people. As soon as the food is received by the food bank, it must be promptly allocated to one of the charities, even before they know what other donations will be received throughout the day. More specifically, 
the online model for the donated food allocation is as follows:
suppose there are $n$ agents and $m$ items. Each agent $a_i$ has some utility $u_{ij}$ for each item $j$. At time moment $j$, item $j$ arrives, and each agent $a_i$ bids for $j$ and an \emph{online randomized allocation mechanism} allocates this item to an agent.  
\begin{itemize}
    \item \texttt{LIKE mechanism}: agent $a_i$ is feasible for item $j$ if they bid positively for it;
    \item \texttt{Balanced LIKE machanism}: agent $a_i$ is feasible for item $j$ if they bid positively for it and have fewest items among those bidding positively for it. 
\end{itemize}
Nevertheless, there are limitations to the study as it only examines a basic set of bidding mechanisms in which agents express their preference for items using binary utilities (0/1).
However, what if agents have arbitrary entitlements instead of equal entitlements?
For example, it might happen that a charity only serves $20\%$ of the population of poor people and another charity serves more than $50\%$ of the population of poor people. In that case, the former charity should receive around $20\%$ of the total food allocated to the food bank, and the latter charity should receive more than $50\%$ of the total food allocated to the food bank.  
Our study is driven by the desire to explore a scenario in which charities possess arbitrary entitlements regarding the food they can acquire from the food bank in the online setting.
\subsection{Our Contributions}
Our project is inspired by the previous work of Aleksandrov et al \cite{aleksandrov2015online}. 
The setting in this paper is: there is a set of $n$ agents $N = \{ A_1, \dots, A_n\}$ and a set of $m$ items $\mathcal{F} = \{F_1,\dots, F_m\}$. Each agent $A_j$ has a utility $u_j(F_i)$ over the item $F_i$. Each item $F_i$ arrives at each time step $t_k$ and each agent bids (i.e. reports a value) for that item. Based on their bidding the item must be assigned to one of the agents. The next item is then revealed. This continues for $m$ steps. To allocate items in this online model, the authors consider a simple class of bidding mechanisms in which agents merely declare if they like items or not (i.e. they assign either $0$ or $1$ value to each item). 

In our project, we study the online food allocation model where agents have arbitrary entitlements. We extend this model to the setting where agents have arbitrary entitlements instead of equal entitlements. In our setting, in addition to the previous setting, each agent $A_i \in N$ has a fixed weight $w_i \in \mathbb{R}_{>0}$; these weights regulate how agents value their own allocated bundles relative to those of other agents. 
We have formulated the following research inquiries:
\begin{enumerate}
    \item  Is it feasible to adapt and expand the mechanisms detailed in \cite{aleksandrov2015online} for online food distribution with equal entitlements to accommodate situations where the agents possess varying entitlements?
    \item To what extent do the proposed algorithms adhere to fundamental principles such as \emph{Strategyproofness} and \emph{Envy-freeness}?
\end{enumerate}
To tackle these questions, we introduced two online weighted fair division mechanisms:
\begin{itemize}
    \item \texttt{Weighted LIKE} mechanism: Assign the food item $F_i$ to agent $j$ with probability $\frac{w_i}{\sum_{i \in B_j} w_j}$;
    \item \texttt{Weighted Balanced LIKE} mechanism: Assign food item $F_i$ at round $i$ uniformly at random to agent $j$ who has the least $\frac{A_{i-1, j}}{w_j}$ value.
\end{itemize}

\subsection{Outline of the Paper}
Our report is organized as follows: In Section \ref{sec: motivation}, we provide the motivation behind our report. In Section \ref{sec: preliminaries}, we provide the necessary notations and preliminaries used in our proofs. Section \ref{sec:online0/1} briefly mentions the results presented in \cite{aleksandrov2015online}, while also defining the notions of \emph{strategy-proofness} and \emph{envy-freeness}. Next, in Section \ref{sec:our_work}, we present our \texttt{weighted LIKE} and \texttt{weighted Balanced LIKE} mechanisms, along with a table summarizing our results. We also compare our findings with the previously known results of Alekasandrov et al \cite{aleksandrov2015online}. We introduce the concept of the game tree, which helps in calculating expected utility in our \texttt{weighted Balanced LIKE} mechanism, in Section \ref{sec:gametree}. In Section \ref{sec:strategy-proof}, we present our results related to the \emph{strategy-proofness} of the \texttt{weighted LIKE} and \texttt{weighted Balanced LIKE} mechanisms. Finally, in Section \ref{sec:envy-free}, we present our results related to the \emph{envy-freeness} of the \texttt{weighted LIKE} and \texttt{weighted Balanced LIKE} mechanisms.



\section{Notations and Preliminaries}\label{sec: preliminaries}
The study considers a scenario in which there are $N$ agents denoted by the set ${1, 2, . . . , n}$. The agents interact over a set of $m$ items denoted by $\mathcal{F} = \{F_1, F_2, \dots, F_m\}$. The entitlement of agent $j$ for a particular item $F_i$ is denoted by $w_j$, while $B_k$ represents the set of agents who have placed bids for item $F_k$ during round $k$. The utility of agent $j$ for the item $F_i$ is represented by $u_j(F_i)$, and the probability that agent $j$ receives the item $F_i$ is denoted by $P(F_i,j)$. The set of items allocated to agent $j$ after the last round is represented by $A_j$, and the expected utility of agent $i$ over the allocation assigned to agent $j$ is denoted by $E[u_i(A_j)]$. Furthermore, the study defines $A_{j,k}$ as the set of items allocated to agent $j$ up until round $k$, while \emph{received-states$(j, k)$} is the set of states in the game tree that signify that agent $j$ received item $F_k$ in round $k$. Additionally, the \emph{paths$(j,k)$} denote the set of paths from the \emph{received-states$(j,k)$} to the root in the game tree.

\section{Online Model with 0/1 Utilities and equal entitlements}\label{sec:online0/1}
In \cite{aleksandrov2015online} the authors devised two mechanisms (LIKE and Balanced LIKE) for fairly allocating the food items to charities. 
In this section, we mention the examples for negative results provided in \cite{aleksandrov2015online}, which are applicable to our setting as well.



\subsection{Fair-allocation Mechanisms}


\begin{itemize}
    \item \texttt{LIKE mechanism}: Once a food packet arrives at the food bank, the food bank asks the charities whether they would like to have the food packet or not. Then the food packet is allocated to one of those charities (chosen uniformly at random) that like to have that food packet. 
    \item \texttt{Balanced LIKE} mechanism: Suppose a food packet $F_k$ arrives at a time step $k$ at the food bank. This time also the food bank asks the charities whether they would like to have the food packet or not. In the \texttt{Balanced LIKE} mechanism, the food packet is allocated to one of those charities (chosen uniformly at random) that like to have that food packet and has the least number of food packets until time step $k - 1$. 
\end{itemize}

\subsection{Justice critera}
We now examine the axiomatic properties of these mechanisms.

\subsubsection{Strategy-proofness}
The \emph{strategy-proofness} property ensures that no participant in the mechanism can benefit from misrepresenting their preferences or manipulating the mechanism in any way. The goal of \emph{strategy-proofness} is to ensure that the outcome of the allocation process is fair and efficient, regardless of how participants behave strategically. In this Section, we briefly discuss the \emph{strategy-proofness} results discussed in \cite{aleksandrov2015online}.

First, we mention the positive results given in \cite{aleksandrov2015online} related to \emph{strategy-proofness} of the \texttt{LIKE} and \texttt{Balance LIKE} mechanisms.
\begin{theorem}\label{thm:like_str_proof}
\texttt{LIKE mechanism} is strategy-proof.
\end{theorem}

\begin{theorem}\label{thm:blike_str_proof_0/1_2agents}
  The \texttt{Balanced-LIKE} mechanism is strategy-proof for two agents with $0/1$ utilities.    
\end{theorem}


Now we mention the negative results given in \cite{aleksandrov2015online} related to strategy proofness of \texttt{LIKE} and \texttt{Balanced LIKE} mechanisms. We also mention the examples here, the same examples would work for the negative results in our setting with arbitrary entitlements.
\begin{theorem}\label{thm:blike_str_proof_0/1}
The \texttt{Balanced-LIKE} mechanism is \emph{not} strategy-proof for more than two agents even with $0/1$ utilities.
\end{theorem}

\begin{proof}
    Suppose, there are three agents $A$, $B$ and $C$ and three items $I_1$, $I_2$ and $I_3$. Suppose, Agent $A$ has utility $1$ for every item. Agent $B$ has utility $1$ items $I_1$ and $I_3$ and $0$ for $I_2$. Agent $3$ has utility $1$ only for item $I_2$ and $0$ for the rest. Now, if agent $1$ sincerely bids then he gets an expected utility of $\frac{9}{8}$ whereas if she only bids for the items $I_2$ and $I_3$ then she gets an expected utility of $\frac{5}{4}$.
\end{proof}

\begin{theorem}\label{thm:blike_str_proof_general}\label{thm:bal_like_strproof_genu}
 The \texttt{Balanced-LIKE} mechanism is \emph{not} strategy-proof for two agents with general utilities.
 \end{theorem}

\begin{proof}
Suppose, there are two agents $A$ and $B$ and two items $I_1$ and $I_2$. Suppose, agent $A$ has utilities $1/2$ for both the items $I_1$ and $I_2$. Agent $B$ has utility $1/4$ for the item $I_1$ and $3/4$ for the item $I_2$. Then if agent $B$ sincerely bids then she gets an expected utility of $1/2$ however if she bids only for the item $2$ she will get an expected utility of $3/4$.     
\end{proof}

\subsubsection{Envy-freeness}
We take into account both \emph{ex-post} and \emph{ex-ante} fairness notions, where \emph{Ex-post fairness} refers to the fairness of an outcome after the resources have been allocated. It judges the fairness of an allocation by looking at the end result, regardless of how the resources were distributed.
\textcolor{black}{\emph{Ex-ante fairness}}, on the other hand, refers to the fairness of the allocation process before it takes place. \textcolor{black}{It focuses on the fairness of the rules and procedures used to allocate resources}, rather than the outcome of the allocation.

\emph{Ex-post} fairness ensures that the outcome of the allocation is just and equitable, while \emph{ex-ante} fairness ensures that the allocation process is fair and unbiased. 
By considering both notions, we can create fair allocation mechanisms that produce just and equitable outcomes while ensuring that the process itself is fair and unbiased.

Now, we state the positive results mentioned in \cite{aleksandrov2015online} for the \texttt{LIKE} and \texttt{Balanced LIKE} mechanisms.

\begin{theorem}
    The \texttt{LIKE mechanism} is envy-free ex-ante.
\end{theorem}

\begin{theorem}
    The \texttt{Balanced LIKE} mechanism is both envy-free ex-ante and bounded envy-free ex-post.
\end{theorem}

Now, we mention the negative results related to envy-freeness for the \texttt{LIKE} and \texttt{Balanced LIKE} mechanisms given in \cite{aleksandrov2015online}. We also give examples of the negative results here. The same examples would work for the negative results in our setting with arbitrary entitlements.

\begin{theorem}
    \texttt{LIKE mechanism} is not bounded envy-free ex-post.
\end{theorem}

\begin{proof}
    To show that the \texttt{LIKE mechanism} is not bounded envy-free ex-post even with 0/1 utilities, suppose 2 agents have utility $1$ for all $m$ items. There is one outcome in which the first agent gets lucky and is assigned every item. However, in this case, the other agent assigns utility of $m$ units greater to the first agent’s allocation than to their own (empty) allocation.
\end{proof}

\begin{theorem}
    The \texttt{Balanced LIKE} mechanism is neither bounded envy-free ex-ante nor bounded envy-free ex-post with general utilities.
\end{theorem}

\begin{proof}
    Consider $2$ agents and $2$ items, $a$ and $b$. Suppose agent $1$ has utility $0$ for $a$ and $p$ for $b$, but agent $2$ has utility $1$ for item $a$ and $(p - 1)$ for item $b$ where $p > 2$. Note that both agents have the same sum of utilities for the two items. If agents bid sincerely then agent $2$ gets an expected utility of just $1$ and envies ex-ante agent $1$’s allocation which gives agent $2$ an expected utility of $(p - 1)$. As $p$ is unbounded, agent $2$ does not have bounded envy ex-post or ex-ante of agent $1$.
\end{proof}

\section{Online fair division with 0/1 utilities and arbitrary entitlements} \label{sec:our_work}

In our scenario, there are $n$ agents, and each agent has a private utility for $m$ items. At each time step, one of the $m$ items is presented, and the allocation mechanism must assign it to one of the agents. This process repeats for $m$ steps. To allocate items in this online model, we utilize a basic set of bidding mechanisms in which agents simply indicate their preference for or against an item. Each agent $A_i$ also has a weight $w_i$ denoting its total entitlement.

We propose the \texttt{Weighted LIKE} mechanism. Let after an item arrives, the agents who bid for it be denoted by the set $B \subseteq \{1,2,\dots,n\}$. Then by this mechanism, each agent $i \in B$ gets the item with probability $\frac{w_i}{\sum_{j \in B}{w_j}}$, where $w_i$ is the weight of the agent $i$.

The \texttt{Weighted LIKE} mechanism has a flaw in that agents may encounter misfortune, leading to unsuccessful bids for all items due to unfavorable coin tosses. This outcome is particularly undesirable in our Food Bank context, where an entire segment of the population may go hungry for the night. As a result, we explore a slightly more advanced mechanism to address this issue.

The \texttt{Weighted Balanced LIKE} mechanism assigns a food item $F_i$ at time step $i$ to an agent $j$ which has the least value of items allocated until round $(i - 1)$ relative to its weight, $\frac{|A_{j,(i - 1)}|}{w_j}$. If multiple items have an equal score, then it is distributed uniformly at random between them. This mechanism is less likely to leave agents empty-handed than the Weighted LIKE mechanism.

We provide Table \ref{table:1}, which summarises all our results for \texttt{weighted LIKE} and \texttt{weighted Balanced LIKE} mechanisms. The results we got are exactly the same as the results obtained by Aleksandrov et al \cite{aleksandrov2015online} for LIKE and Balanced LIKE mechanisms. Our negative results use the same examples as Aleksandrov et al \cite{aleksandrov2015online} because our \texttt{weighted LIKE} and \texttt{weighted Balanced LIKE} mechanisms behave as LIKE and Balanced LIKE mechanisms when there are equal entitlements.   

\begin{table*}[tbt]
\centering
\begin{tabular}{c c c c c} 
 \hline
 Method & Condition & Strategy Proof & BEF Ex-Post & BEF Ex-Ante \\ 
 \hline
 $\mathsf{W.LIKE}$  & 0/1-Utility & yes & yes & no \\ 
 $\mathsf{W.BALANCED}$-$\mathsf{LIKE}$ & General Utility & yes & yes & no \\ 
 $\mathsf{W.LIKE}$ & 0/1-Utility & no for $> 2$ agents & yes & yes \\
$\mathsf{W.BALANCED}$-$\mathsf{LIKE}$ & General Utility & no & no & no \\ 
 \hline
\end{tabular}
\caption{Table summarizing the fairness properties of \emph{WEIGHTED LIKE} and \emph{WEIGHTED BALANCED LIKE} mechanism}
\label{table:1}
\end{table*}

\section{Game Tree Construction}\label{sec:gametree}

To calculate the expected utility of an agent $j$ in a weighted balanced LIKE mechanism we need to take the help of \emph{game tree}.

A game tree is a tree where each node represents the game's state. For the weighted balanced LIKE mechanism, each node is labeled with a tuple $(k, \frac{|A_{1,k}|}{w_1} \dots \frac{|A_{n,k}|}{w_n})$ where $|A_{j,k}|$ represents the number of items allocated to agent $j$ until round $k$ and $w_j$ is the entitlement of agent $j$. Each node has $(n + 1)$ children where the first $n$ children represent one of $n$ agents that received the item on that round and the $(n + 1)th$ node represents that none of the agents received the item that round. We put weights on the edges connecting the parent node and the child node which represent the probability the game can go to the child state from the parent state.

Suppose at height $h$ (i.e. round $h$ of the game) of the game tree \emph{received-states$(j,h)$} represents the set of states where agent $j$ received the food item $F_h$ at round $h$. Note that at height $(h - 1)$ there are $(n + 1)^{h - 1}$ nodes in the game tree. Now, for each of these nodes at height $(h - 1)$ we will get exactly one child which represents that agent $j$ received the food item $F_h$ at round $h$. Hence, there are $(n + 1)^{h - 1}$ \emph{received-states$(j,h)$} at height $h$. We define \emph{paths$(j, h)$} as the set of paths from this \emph{received-states$(j,h)$} at height $h$ to the root of the game tree. For $0/1$ utilities, the expected utility got by agent $j$ (assuming he bid honestly) at round $h$ for the food item, $F_h$ is the sum of the products of weights along the \emph{paths$(j,h)$}. Refer to Figure \ref{fig:game_tree} for a better understanding of the game tree.

\begin{figure}
    \centering
    \includegraphics[scale = 0.25]{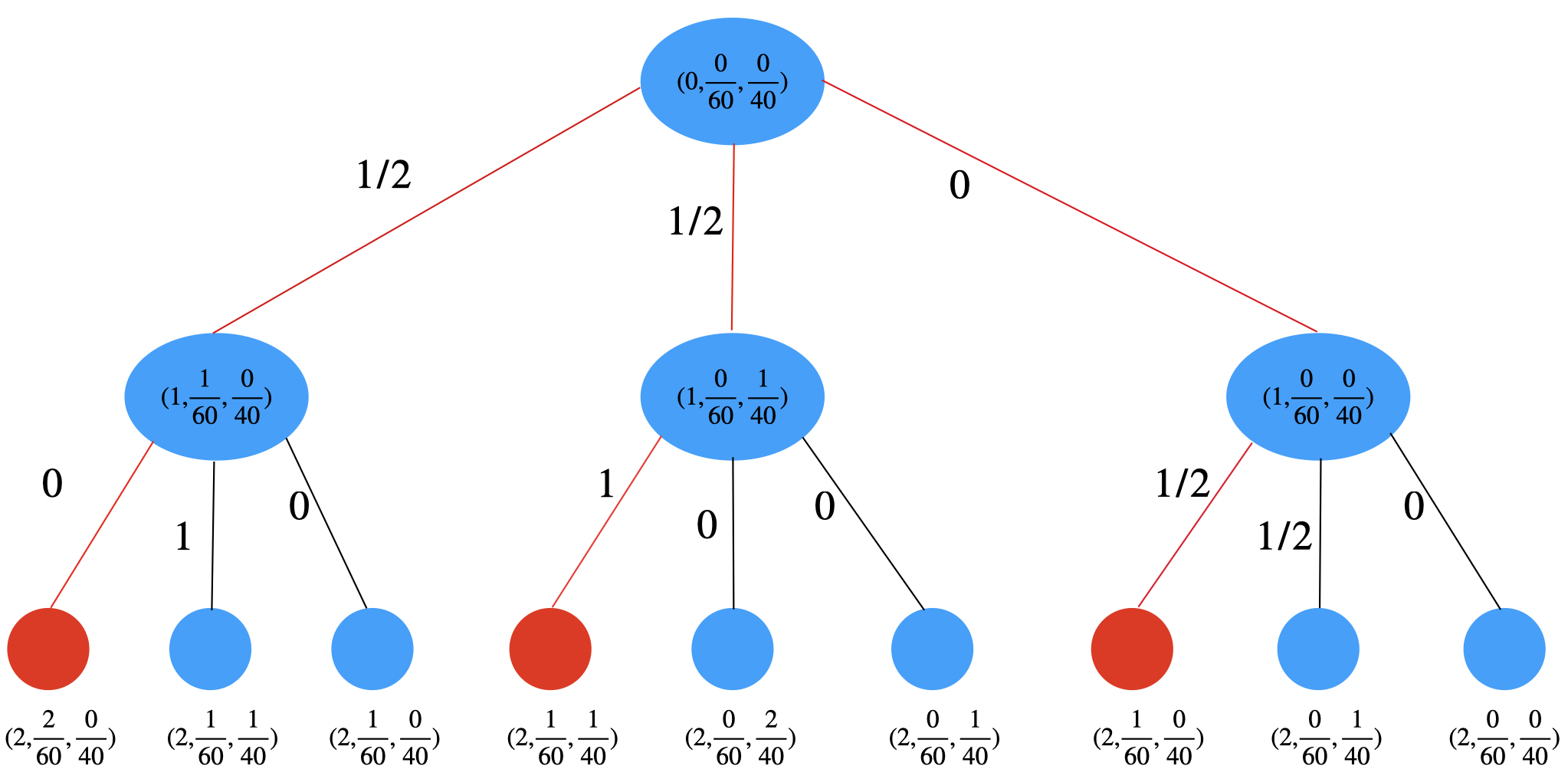}
    \caption{Image of a game tree for 2 agents and 2 items. We assume all agents bid for all items. Each node has three children, the first child represents that agent 1 received the item, the second child represents the second agent received the item and the third child represents that none of them received the item $F_h$ at height $h$. The weights represent the probability the game can reach the child state from the parent state using the weighted balanced LIKE mechanism. At height $2$ of the game tree there are $3$ states representing \emph{received-states$(1,2)$} (Red color states) and the red color paths are \emph{paths(1,2)}. Hence, the expected utility for Agent $1$ for item $F_2$ is $(0 * 1/2 + 1 * 1/2 + 1/2 * 0) = 1/2$ }
    \label{fig:game_tree}
\end{figure}
\section{Strategy-proofness}\label{sec:strategy-proof}

\begin{theorem}
    The weighted LIKE mechanism is strategy-proof even with general utilities.
\end{theorem}

\begin{proof}
    Case 1: Agent $j$ bids for item $F_i$ but her actual bid for item $F_i$ was $0$: In this case, the expected utility received by agent $j$ remains the same (w.r.t. to the case where she would have bidden truthfully) because her actual utility for item $F_i$ was $0$. So, she does not have any incentive to bid for item $F_i$.

    Case 2: Agent $j$ does bid for item $F_i$ but her actual utility for item $F_i$ was $u > 0$: In this case, the expected utility received by agent $j$ decreases (w.r.t. to the case where she would have bidden truthfully) because her actual utility for item $i$ was $u > 0$. If she does not bid for item $F_i$ then the probability of receiving item $F_i$ is $0$ and it also does not increase the probability of receiving the future items in the weighted LIKE mechanism (as compared to the weighted Balanced LIKE mechanism). So, she does not have any incentive to not bid for item $F_i$. 
\end{proof}
\begin{theorem}
    The weighted balanced LIKE mechanism is strategy-proof for $2$ agents with $0/1$ utilities.
\end{theorem}

\begin{proof}
    When there are $2$ agents, we prove that the weighted balanced LIKE is strategy-proof for any agent independent of its weight. WLOG let us try to show that it is strategy-proof for agent $1$.
    
    Case 1: Suppose the actual utility for an item $F_i$ for agent $1$ is $0$ but she bids $1$ for that item: The expected utility got by agent $1$ in the weighted balanced LIKE mechanism is $\sum_{k \in \mathcal{F}}P(F_k,1) \times u_1(F_i)$. Now, since the actual utility for the item $F_i$ is $0$ i.e. $u_1(F_i) = 0$, bidding $1$ for the item $F_i$ won't increase her expected utility. Hence, she has no incentive to bid $1$ for the item in this case.

    Case 2: Suppose the actual utility for an item $F_i$ for agent $1$ is $1$ but she bids $0$ for that item: This case again has two sub-cases.

    Case 2(a): Agent $2$ has utility $0$ for item $F_i$: Since Agent $1$ is bidding $0$ for $F_i$, it turns out that both the agents have utility $0$ for item $F_i$. In this case, we can simply disregard the item $F_i$ as if it never existed because we are not going to allocate $F_i$ to anyone in this case. Hence, in this case, agent $1$ does not have any incentive to bid $0$ for this item.

    Case 2(b): Agent $2$ has utility $1$ for item $F_i$: Suppose Agent $1$ lies at the $i^{th}$ round. Since his actual bid is $1$ and he is bidding $0$ so, he is going to lose some non-zero utility, $u$ (say) for the item $F_i$. This sacrifice of agent $1$ would be fruitful if he gains more than the utility of $u$ in future allocations. 

    Now, at the $ith$ round Agent $1$ lies her bid to $0$. This means that the weight of the edges in between the nodes in the $(i - 1)th$ level in the game tree and its child node which represents agent $1$ received the item in the $ith$ round would turn to $0$. Similarly, the weight of the edges in between the nodes in the $(i - 1)th$ level in the game tree and its child node which represents agent $2$ received the item in the $ith$ round would turn to $1$.

    Consider one of these nodes at level $(i - 1)$, $(i - 1, \frac{|A_{1,(i - 1)}|}{w_1}, \frac{|A_{2,(i - 1)}|}{w_2})$. Suppose with probability $p$ it goes to the state $(i, \frac{|A_{1,(i - 1)}| + 1}{w_1}, \frac{|A_{2,(i - 1)}|}{w_2})$, (this state represents Agent $1$ received the item) and with probability $(1 - p)$ it goes to the state $(i, \frac{|A_{1,(i - 1)}|}{w_1}, \frac{|A_{2,(i - 1)}| + 1}{w_2})$, (this state represents Agent $2$ received the item). Now, since agent $1$ is lying, the probability $p$ would change to $0$, and the probability $(1 - p)$ would change to $1$. Hence, she is losing a utility of $p$ because the game goes to state $(i, \frac{|A_{1,(i - 1)}| + 1}{w_1}, \frac{|A_{2,(i - 1)}|}{w_2})$ with probability $p$. Suppose at level $i'$ agent $1$ gains the utility which she lost due to lying at round $i$. In the best case from level $i$ to $i'$ she should have received the items at every round with probability $1$. Hence, the product of the weights along the path from the leaf present in the sub-tree rooted at $(i, \frac{|A_{1,(i - 1)}|}{w_1}, \frac{|A_{2,(i - 1)}| + 1}{w_2})$ at level $i'$ to the root of the game tree is $(1 - p)$ when she is telling the truth and $1$ when she is lying. Hence, the gain in her utility is at most $(1 - (1 - p)) = p$. She won't gain utilities for any more items in the future because at round $i'$ the sub-tree rooted at $(i, \frac{|A_{1,(i - 1)}| + 1}{w_1}, \frac{|A_{2,(i - 1)}|}{w_2})$ reaches the state with probability $1$ which has the same label as the state reached under the sub-tree rooted at $(i, \frac{|A_{1,(i - 1)}|}{w_1}, \frac{|A_{2,(i - 1)}| + 1}{w_2})$ at round $i'$. Refer to Figure ~\ref{fig:strategyproof} for a better understanding.   
\end{proof}

\begin{figure}
    \centering
    \includegraphics[scale = 0.25]{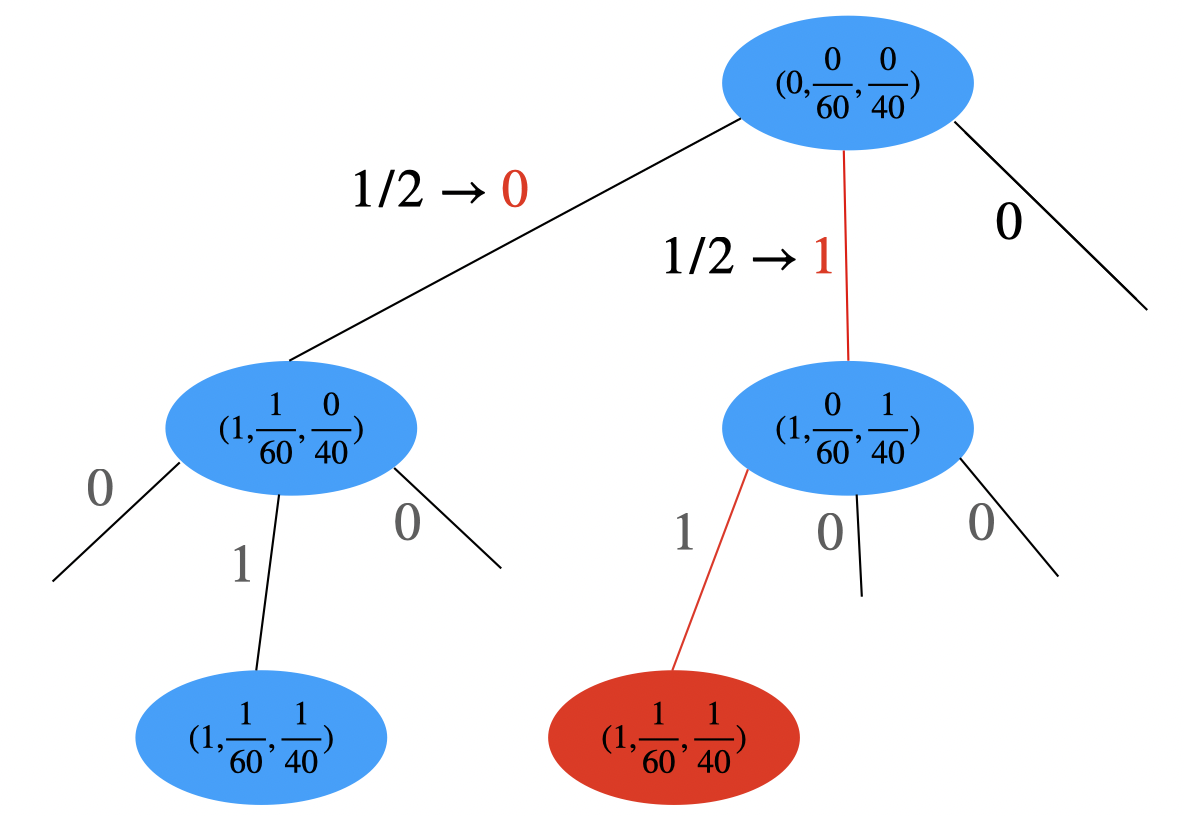}
    \caption{We have two agents $A_1$ and $A_2$ and $m$ items. Suppose agents $A_1$ and $A_2$ have utilities $1$ for the first two items and arbitrary utilities for the remaining items. Agent $1$ has lied for the first item so his probability of getting the first item changes from $1/2$ to $0$ and agent $2$ probability of getting the first item changes from $1/2$ to $1$. Agent $1$ loses a utility of $1/2$ for the first item but gains a utility of $1/2$ for the second item, a product of the weights along the red color path from the node $(1, 1/60, 1/40)$ (red color node) to $(0, 0/60, 0/40)$. He cannot gain any utility for any items in the future because at height $2$ the game has reached the states which have the same labels $(1, 1/60, 1/40)$ hence the sub-trees rooted under them would be the same.}
    \label{fig:strategyproof}
\end{figure}
\section{Envy-freeness}\label{sec:envy-free}
\begin{theorem}
    The weighted LIKE mechanism is envy-free ex-ante
\end{theorem}

\begin{proof}
    So, for any two agents $i$ and $j$, we need to prove\\\\
    $\frac{E[u_i(A_j)]}{w_j} \leq \frac{E[u_i(A_i)]}{w_i}$ where $E[u_i(A_j)]$ is the expected utility of agent $i$ over the allocation assigned to agent $j$.\\\\
    $\Leftrightarrow$ $\frac{\sum_{k = 1}^{m} P(F_k, j) \times u_i(F_k)}{w_j} \leq \frac{\sum_{k = 1}^{m} P(F_k, i) \times u_i(F_k)}{w_i}$\\\\
    $\Leftrightarrow$ $\sum_{k = 1}^{m}\frac{P(F_k, j) \times u_i(F_k)}{w_j} \leq \sum_{k = 1}^{m}\frac{P(F_k, i) \times u_i(F_k)}{w_i}$ $\dots (1)$\\\\
    Now,\\\\
    $P(F_k,j) = \frac{w_j}{\sum_{l \in B_k}w_l}$ $\dots (2)$\\\\
    $P(F_k,i) = \frac{w_i}{\sum_{l \in B_k}w_l}$ $\dots (3)$ where $B_k$ is the set of agents who bid for the item $F_k$ in round $k$.\\\\
    Now substituting the values of $(2)$ and $(3)$ in $(1)$ we get

    $\sum_{k = 1}^{m}\frac{P(F_k, j) \times u_i(F_k)}{w_j} = \sum_{k = 1}^{m}\frac{P(F_k, i) \times u_i(F_k)}{w_i}$\\\\
    $\Leftrightarrow$ $\frac{E[u_i(A_j)]}{w_j} = \frac{E[u_i(A_i)]}{w_i}$ (Proved)
\end{proof}
\begin{theorem}
    The weighted balanced LIKE mechanism is bounded envy-free ex-post.
\end{theorem}

\begin{proof}
    Here for any two agents $i$ and $j$, we need to prove $\frac{u_i(A_j)}{w_j} \leq \frac{u_i(A_i)}{w_i} + 1$. Since, we have $0/1$ utilities $u_i(A_j) = |A_j|$ and $u_i(A_i) = |A_i|$. So we need to prove $\frac{|A_j|}{w_j} \leq \frac{|A_i|}{w_i} + 1$. Now, suppose $\frac{|A_j|}{w_j} = \frac{|A_i|}{w_i}$ when $|A_j| = n_1$ and $|A_i| = n_2$ thus $\frac{n_1}{w_j} = \frac{n_2}{w_i}$. Now, the weighted Balanced LIKE mechanism allocates the next item that has the least number of items allocated until the previous round relative to its weight. Hence, in the worst case, $\frac{|A_j|}{w_j} = \frac{n_1 + 1}{w_j}$ and $\frac{|A_i|}{w_i} = \frac{n_2}{w_j}$. Hence, $\frac{|A_j|}{w_j} \leq \frac{|A_i|}{w_i} + 1$ thus, $\frac{u_i(A_j)}{w_j} \leq \frac{u_i(A_i)}{w_i} + 1$.
\end{proof}

\begin{theorem}
    The weighted balanced LIKE mechanism is envy-free ex-ante.
\end{theorem}

\begin{proof}
    We prove this using induction on the number of items. For any two agents, agent $i$ and agent $j$, agent $i$ can only envy agent $j$ on the items which are present in $A_j$ and for which she has a utility of $1$. Again since the items have been added to $j's$ bundle so $j$ also has a utility of $1$ for those items. Suppose until level $k$ of the game tree, (that is until $kth$ item) the weighted balanced LIKE mechanism is envy-free ex-ante. So, for the $(k + 1)th$ item we need to show that the increase in the expected utility for agent $i$ is at least the increase in the expected utility for agent $j$.  At Section \ref{sec:gametree} we saw that the expected utility of agent $j$ for the item $(k + 1)$ is the sum of the products of the weights along the \emph{paths$(j, (k + 1))$}, same goes for agent $i$. Consider the special case where agents $i$ and $j$ bid for all the items. In that case, note that the sum of the products of the weights along the \emph{paths$(j, (k + 1))$} and the sum of the products of the weights along the \emph{paths$(i, (k + 1))$} would be the same. This is because each node has a child node in the game tree representing agent $j$ got the item in the next level and also a child node representing agent $i$ got the item in the next level. 
    
    For the other case, where either agent $j$ or agent $i$ (or both) is not bidding for some items, the increase in the expected utility balances out between the two agents using our weighted balanced LIKE mechanism. Suppose agent $i$ has not bid for a sequence of items then those items must be added to $j's$ bundle (if not, then we can disregard those items for our analysis of envy-freeness). Now, for the next item for which both of them bid together agent $i$ has a higher probability of receiving that item. So, agent $i$ does not envy agent $j$ because she received the item for which both bid with higher probability, and agent $j$ does not envy agent $i$ because she received more items in the past when agent $i$ did not bid for the previous items.     
\end{proof}

\section{Conclusion}
In conclusion, this study has analyzed the problem of online fair allocation with 0/1 utility mechanisms in \cite{aleksandrov2015online} and introduced two online weighted fair division mechanisms, the \texttt{Weighted LIKE} mechanism, and the \texttt{Weighted Balanced LIKE} mechanism. These mechanisms have been developed as a first attempt to expand and adapt the online allocation with equal entitlements problem to a more general weighted condition. Moreover, we have proven that the proposed algorithms adhere to fundamental principles such as \emph{Strategyproofness} and \emph{Envy-freeness}. The results of this study provide a starting point for future research in the area of online fair allocation with weighted entitlements. In particular, our future work will focus on improving the performance of the proposed mechanisms and extending them to more complex settings with additional constraints and preferences. Overall, the proposed mechanisms can be valuable tools for ensuring fairness in various online allocation scenarios, such as resource allocation, job assignments, and scheduling.

Our results for the weighted LIKE and weighted Balanced LIKE mechanisms match those obtained by Aleksandrov et al. \cite{aleksandrov2015online} perfectly. Despite the fact that the outcomes are the same, our project's proof methodology differs from that of Aleksandrov et al. The authors of the Aleksandrov et al study only offer proof sketches for the theorems pertaining to strategy-proofness and envy-freeness. All of the proofs we did for strategy-proofness and envy-freeness would also apply in their context because our work is an extension of their study. We hope that our project will make it easier for readers to grasp how the theorems in \cite{aleksandrov2015online} are proved.

\bibliographystyle{alpha}
\bibliography{biblio}
\end{document}